\documentclass[conference]{IEEEtran}
\usepackage{ifpdf}
\usepackage{cite}
\usepackage{tikz}
\usepackage{lipsum}
\usepackage{mathtools}
\usepackage{cuted}
\usepackage[left=0.67in,right=0.67in,top=0.7in,bottom=0.7in]{geometry}
\usetikzlibrary{decorations, decorations.text,}
\usepackage{float}
\usepackage{amsmath}
\usepackage{amsthm}
\usepackage{mathtools, cuted}
\usepackage{array}
\usepackage[caption=false,font=footnotesize]{subfig}
\usepackage{textcomp}
\usepackage{url}
\usepackage{paralist,esint,amssymb,booktabs,cases,bm,galois}

\newcommand\Tx[1]{\mathrm{#1}}
\newcommand\Se[1]{\mathcal{#1}}
\newcommand\Db[1]{\mathbb{#1}}

\newcommand\MB[1]{\left[#1\right]}
\newcommand\LB[1]{\{#1\}}
\newcommand\SB[1]{\left(#1\right)}

\newcommand{\RN}[1]{\textup{\uppercase\expandafter{\romannumeral#1}}}
\usepackage{cleveref}
\usetikzlibrary{calc}
\usetikzlibrary{shadings}
\newtheorem{theo}{Theorem}

\newtheorem{lemma}{Lemma}
\newtheorem{exam}{Example}
\newtheorem{defi}{Definition}
\newtheorem{rem}{Remark}
\newtheorem{cons}{Construction}

\newcommand{\cloud}[4]{
    \draw[scale = #1,shift = {(#2,#3)},thick] (-1.6,-0.7) .. controls (-2.3,-1.1)
and (-2.7,0.3) .. (-1.7,0.3) .. controls (-1.6,0.7)
and (-1.2,0.9) .. (-0.8,0.7) .. controls (-0.5,1.5)
and (0.6,1.3) .. (0.7,0.5) .. controls (1.5,0.4)
and (1.2,-1) .. (0.4,-0.6) .. controls (0.2,-1)
and (-0.2,-1) .. (-0.5,-0.7) .. controls (-0.9,-1)
and (-1.3,-1) .. cycle;
\node[scale = #1,shift = {(#2,#3)}] at (-0.35*#1,0) {{\huge{\textbf{Cloud $\textbf{#4}$}}}};
}
\newcommand{\server}[3]{
   \fill[scale = #1,shift = {(#2,#3)},top color=gray!50!gray,bottom color=gray!10,middle color=gray,shading=axis,opacity=0.25] (0,0) circle (2 and 0.5);
   \fill[scale = #1,shift = {(#2,#3)},left color=gray!50!gray,right color=gray!50!gray,middle color=gray!50,shading=axis,opacity=0.25] (2,0) -- (2,6) arc (360:180:2 and 0.5) -- (-2,0) arc (180:360:2 and 0.5);
   \fill[scale = #1,shift = {(#2,#3)},top color=gray!90!,bottom color=gray!2,middle color=gray!30,shading=axis,opacity=0.25] (0,6) circle (2 and 0.5);
   \fill[scale = #1,shift = {(#2,#3)},top color=gray!90!,bottom color=gray!2,middle color=gray!30,shading=axis,opacity=0.25] (0,4.5) circle (2 and 0.5);
   \fill[scale = #1,shift = {(#2,#3)},top color=gray!90!,bottom color=gray!2,middle color=gray!30,shading=axis,opacity=0.25] (0,3) circle (2 and 0.5);
   \fill[scale = #1,shift = {(#2,#3)},top color=gray!90!,bottom color=gray!2,middle color=gray!30,shading=axis,opacity=0.25] (0,1.5) circle (2 and 0.5);
   \draw[scale = #1,shift = {(#2,#3)}] (-2,6) -- (-2,0) arc (180:360:2 and 0.5) -- (2,6) ++ (-2,0) circle (2 and 0.5);
   \draw[scale = #1,shift = {(#2,#3)}] (-2,4.5) -- (-2,0) arc (180:360:2 and 0.5) -- (2,4.5) ++ (-2,0) circle (2 and 0.5);
   \draw[scale = #1,shift = {(#2,#3)}] (-2,3) -- (-2,0) arc (180:360:2 and 0.5) -- (2,3) ++ (-2,0) circle (2 and 0.5);
   \draw[scale = #1,shift = {(#2,#3)}] (-2,1.5) -- (-2,0) arc (180:360:2 and 0.5) -- (2,1.5) ++ (-2,0) circle (2 and 0.5);
   \draw[scale = #1,shift = {(#2,#3)}] (-2,0) arc (180:0:2 and 0.5);
}

\usepackage{algorithm}
\usepackage{algpseudocode}
\makeatletter
\def\BState{\State\hskip-\ALG@thistlm}
\makeatother

\errorcontextlines\maxdimen
\makeatletter
\newcommand*{\algrule}[1][\algorithmicindent]{\makebox[#1][l]{\hspace*{.5em}\vrule height 0.9 \baselineskip depth 0.3\baselineskip}}%

\newcount\ALG@printindent@tempcnta
\def\ALG@printindent{%
    \ifnum \theALG@nested>0
        \ifx\ALG@text\ALG@x@notext
            \addvspace{-3pt}
        \else
            \unskip
            \ALG@printindent@tempcnta=1
            \loop
                \algrule[\csname ALG@ind@\the\ALG@printindent@tempcnta\endcsname]%
                \advance \ALG@printindent@tempcnta 1
            \ifnum \ALG@printindent@tempcnta<\numexpr\theALG@nested+1\relax
            \repeat
        \fi
    \fi
    }%
\usepackage{etoolbox}
\patchcmd{\ALG@doentity}{\noindent\hskip\ALG@tlm}{\ALG@printindent}{}{\errmessage{failed to patch}}
\makeatother



\begin{document}
\title{{Hierarchical Coding to Enable Scalability and Flexibility in Heterogeneous Cloud Storage}}
\author{\IEEEauthorblockN{Siyi Yang$^1$, Ahmed Hareedy$^2$, Robert Calderbank$^2$, and Lara Dolecek$^1$}
\IEEEauthorblockA{$^1$ Electrical and Computer Engineering Department, University of California, Los Angeles, Los Angeles, CA 90095 USA\\
$^2$ Electrical and Computer Engineering Department, Duke University, Durham, NC 27705 USA\\
siyiyang@ucla.edu, ahmed.hareedy@duke.edu, robert.calderbank@duke.edu, and dolecek@ee.ucla.edu
}}
\maketitle

\begin{abstract} In order to accommodate the ever-growing data from various, possibly independent, sources and the dynamic nature of data usage rates in practical applications, modern cloud data storage systems are required to be scalable, flexible, and heterogeneous. Codes with hierarchical locality have been intensively studied due to their effectiveness in reducing the average reading time in cloud storage. In this paper, we present the first codes with hierarchical locality that achieve scalability and flexibility in heterogeneous cloud storage using small field size. We propose a double-level construction utilizing so-called Cauchy Reed-Solomon codes. We then develop a triple-level construction based on this double-level code; this construction can be easily generalized into any hierarchical structure with a greater number of layers since it naturally achieves scalability in the cloud storage systems. 
\end{abstract}

\IEEEpeerreviewmaketitle

\section{Introduction}
\label{section: introduction}
Codes offering hierarchical locality have been intensely studied because of their ability to reduce the average reading time in various erasure-resilient data storage applications including Flash storage, redundant array of independent disks (RAID) storage, cloud storage, etc. \cite{huang2017multi,ballentine2018codes,cassuto2017multi}. Codes with shorter block lengths offer lower latency, but they provide limited erasure-correction capability in a cloud storage system. To deal with more erasures, longer codes can be employed. However, since a simultaneous occurrence of a large number of erasures is a rare event, longer codes result in unnecessary extra reading cost, and are on average inefficient. Therefore, maintaining low latency while simultaneously recovering from a potentially large number of erasures is one of the major challenges in cloud storage. Codes with hierarchical locality have been shown to address this issue by providing multi-level access in cloud storage, which enables the data to be read through a chain of network components with increasing data lengths from top to bottom; this architecture is exploited to increase the overall erasure-correction capability\cite{hassner2001integrated}. 



In the literature, codes offering double-level access have been intensely studied\cite{cassuto2017multi,hassner2001integrated,blaum2018extended,zhang2018generalized,wu2017generalized,martnez2018universal}; these codes are applicable in double-level cloud storage. In this configuration, $p$ consecutive local messages are jointly encoded into $p$ correlated local codewords. Each local codeword is stored at the neighboring servers of the corresponding local cloud. The codes are designed such that each local message can be successfully decoded from the corresponding local codeword when there are fewer than $d_1$ local erasures, and the global codeword provides extra protection against $(d_2-d_1)$ unexpected errors in a local codeword, for some $d_2>d_1$. An example having $p=4$ is in Fig.~1. Suppose $d_1=2$ and $d_2=3$. When there is at most $1$ server failure,  accessing the servers connected to cloud $1$ is sufficient to successfully decode the data stored in cloud $1$. If the number of server failures in cloud $1$ is $2$, the data can still be obtained through accessing all the servers. Codes with hierarchical locality are a generalized extension of double-level accessible codes, in which more than two levels of access are allowed and are naturally suitable for cloud storage with multiple layers.

\begin{figure}\label{fig: cloudstorage}
\centering
\begin{tikzpicture}
\cloud{0.4}{0}{0}{};
\foreach \x in {0.4}{
    \foreach \y in {3.5}{
    \cloud{0.3}{\x}{\y}{3};  
    \draw[thick] (0.1,0.8) -- (0.3*\x-0.2,0.3*\y-0.6);
    \foreach \a in {0.1}{
        \foreach \b in {1.0}{
        \server{0.1}{3*\x+10*\a}{3*\y+10*\b};
        \draw[thick] (0.3*\x+0.35*\a,0.3*\y+0.35*\b) -- (0.3*\x+0.95*\a,0.3*\y+0.95*\b);
        }
    }
    \foreach \a in {1}{
        \foreach \b in {0.8}{
        \server{0.1}{3*\x+10*\a}{3*\y+10*\b};
        \draw[thick] (0.3*\x+0.25*\a,0.3*\y+0.2*\b) -- (0.3*\x+0.95*\a,0.3*\y+0.95*\b);
        }
    }
    \foreach \a in {-1.0}{
        \foreach \b in {0.7}{
        \server{0.1}{3*\x+10*\a}{3*\y+10*\b};
        \draw[thick] (0.3*\x+0.35*\a,0.3*\y+0.35*\b) -- (0.3*\x+0.95*\a,0.3*\y+0.95*\b);
        }
    }
    } 
}

\foreach \x in {-5}{
    \foreach \y in {2}{
    \cloud{0.3}{\x}{\y}{2};  
    \draw[thick] (-1.2,0.45) -- (0.3*\x+0.85,0.3*\y-0.4);
    \foreach \a in {-0.1}{
        \foreach \b in {0.8}{
        \server{0.1}{3*\x+10*\a}{3*\y+10*\b};
        \draw[thick] (0.3*\x+0.4*\a,0.3*\y+0.47*\b) -- (0.3*\x+0.95*\a,0.3*\y+0.95*\b);
        }
    }
    \foreach \a in {-0.7}{
        \foreach \b in {-1.2}{
        \server{0.1}{3*\x+10*\a}{3*\y+10*\b};
        \draw[thick] (0.3*\x+0.15*\a,0.3*\y+0.15*\b) -- (0.3*\x+0.95*\a,0.3*\y+0.95*\b+0.6);
        }
    }
    \foreach \a in {-1.2}{
        \foreach \b in {0.3}{
        \server{0.1}{3*\x+10*\a}{3*\y+10*\b};
        \draw[thick] (0.3*\x+0.42*\a,0.3*\y+0.5*\b) -- (0.3*\x+0.9*\a,0.3*\y+0.9*\b);
        }
    }
    } 
}

\foreach \x in {5}{
    \foreach \y in {0}{
    \cloud{0.3}{\x}{\y}{4};  
    \draw[thick] (0.48,0.0) -- (0.3*\x-0.68,0.3*\y-0.15);
    \foreach \a in {0.4}{
        \foreach \b in {0.8}{
        \server{0.1}{3*\x+10*\a}{3*\y+10*\b};
        \draw[thick] (0.3*\x+0.35*\a,0.3*\y+0.35*\b) -- (0.3*\x+0.95*\a,0.3*\y+0.95*\b);
        }
    }
    \foreach \a in {-0.5}{
        \foreach \b in {-1.2}{
        \server{0.1}{3*\x+10*\a}{3*\y+10*\b};
        \draw[thick] (0.3*\x+0.2*\a,0.3*\y+0.22*\b) -- (0.3*\x+0.95*\a,0.3*\y+0.95*\b+0.6);
        }
    }
    \foreach \a in {0.7}{
        \foreach \b in {-1.2}{
        \server{0.1}{3*\x+10*\a}{3*\y+10*\b};
        \draw[thick] (0.3*\x+0.15*\a,0.3*\y+0.15*\b) -- (0.3*\x+0.95*\a,0.3*\y+0.95*\b+0.6);
        }
    }
    } 
}

\foreach \x in {-2}{
    \foreach \y in {-4}{
    \cloud{0.3}{\x}{\y}{1};  
    \draw[thick] (-0.3,-0.35) -- (0.3*\x,0.3*\y+0.37);
    \foreach \a in {-0.1}{
        \foreach \b in {-1.4}{
        \server{0.1}{3*\x+10*\a}{3*\y+10*\b};
        \draw[thick] (0.3*\x+0.2*\a,0.3*\y+0.2*\b) -- (0.3*\x+0.95*\a,0.3*\y+0.95*\b+0.6);
        }
    }
    \foreach \a in {-1.2}{
        \foreach \b in {-0.5}{
        \server{0.1}{3*\x+10*\a}{3*\y+10*\b};
        \draw[thick] (0.3*\x+0.4*\a,0.3*\y+0.4*\b) -- (0.3*\x+0.85*\a,0.3*\y+0.85*\b);
        }
    }
    \foreach \a in {1}{
        \foreach \b in {-0.8}{
        \server{0.1}{3*\x+10*\a}{3*\y+10*\b};
        \draw[thick] (0.3*\x+0.3*\a,0.3*\y+0.2*\b) -- (0.3*\x+0.8*\a,0.3*\y+0.8*\b);
        }
    }
    } 
}
\end{tikzpicture}
\vspace{-0.5em}
\caption{Double-level cloud storage. Servers connected to local clouds store the local codewords; the local clouds are connected to a central cloud.}
\end{figure}
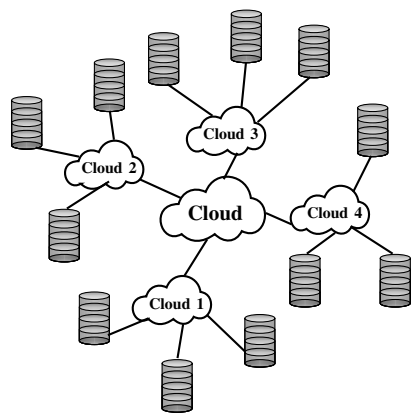

Along with \emph{hierarchical locality} discussed previously, it is also important for the coding schemes to support scalable, heterogeneous, and flexible cloud storage\cite{rimal2009taxonomy}. \emph{Scalability} enables expanding the backbone network to accommodate additional workload, i.e., additional clouds, without rebuilding the entire infrastructure. \emph{Heterogeneity} refers to the property of allowing nonidentical local data lengths and providing unequal local protection, which is important for cloud storage with heterogeneous structures. A heterogeneous structure arises in networks consisting of geographically separated components, and they often store data from different sources. \emph{Flexibility} has been firstly investigated for dynamic data storage systems in \cite{martnez2018universal}, and it refers to the property that the local cloud can be split into two smaller local clouds without worsening the global erasure-correction capability nor changing the remaining components. This splitting, for example, is applied when cold data stored at a local cloud become hot unexpectedly.

Various codes offering hierarchical locality have been studied. Cassuto \emph{et al.}\cite{cassuto2017multi} presented so-called multi-block interleaved codes that provide double-level access; this work introduced the concept of multi-level access. The family of integrated-interleaved (I-I) codes, including generalized integrated interleaved (GII) codes and extended integrated interleaved (EII) codes, has been a major prototype for codes with multi-level access \cite{hassner2001integrated,wu2017generalized,zhang2018generalized,blaum2018extended}. GII codes have the advantage of correcting a large set of error patterns, but the distribution of the data symbols is highly restricted, and all the local codewords are equally protected. EII codes are extensions of GII codes with double-level access, where specific arrangements of data symbols have been investigated, mitigating the aforementioned restriction. However, no similar study has been proposed for GII codes with hierarchical locality. Therefore, I-I codes are more suitable for applications where heterogeneity and flexibility are less important. Sum-rank codes are another family of codes that is proposed for dynamic distributed storage offering double-level access\cite{martnez2018universal}. These codes are maximally recoverable, flexible, and allow unequal protection for local data. However, sum-rank codes require a finite field size that grows exponentially with the maximum local block length, which is a major obstacle to being implemented in real world applications.

In this paper, we introduce code constructions with hierarchical locality and a small field size that achieve scalability, heterogeneity, and flexibility. The paper is organized as follows. In \Cref{section: notation and preliminaries}, we introduce the notation and preliminaries. In \Cref{section: codes for multi-level access}, we present a new construction of codes offering hierarchical locality that is based on Cauchy Reed Solomon (CRS) codes. This construction requires a field size that grows linearly with the maximum local codelength. In \Cref{section: scalability and flexibility in cloud storage}, we then show that our coding scheme is scalable, heterogeneous, and flexible. Finally, we summarize our results in \Cref{section: conclusion}.

\section{Notation and Preliminaries}
\label{section: notation and preliminaries}

Throughout the rest of this paper, $\MB{N}$ refers to $\{1,2,\dots,N\}$, and $\MB{a:b}$ refers to $\{a,a+1,\dots,b\}$. Denote the all zero vector of length $s$ by $\bold{0}_s$. Similarly, the all zero matrix of size $s\times t$ is denoted by $\bold{0}_{s\times t}$. The alphabet field, denoted by $\textup{GF}(q)$, is a Galois field of size $q$, where $q$ is a power of a prime. For a vector $\bold{v}$ of length $n$, $v_i$, $1\leq i\leq n$, represents the $i$-th component of $\bold{v}$, and $\bold{v}\MB{a:b}=(v_{a},\dots,v_b)$. For a matrix $\bold{M}$ of size $a\times b$, $\bold{M}\MB{i_1:i_2,j_1:j_2}$ represents the sub-matrix $\bold{M}'$ of $\bold{M}$ such that $(\bold{M}')_{i-i_1+1,j-j_1+1}=(\bold{M})_{i,j}$, $i\in\MB{i_1:i_2}$, $j\in\MB{j_1:j_2}$. All indices start from $1$.

\subsection{Notation and Definitions}
\label{subsection: notation and definitions}

Let $\bold{m}$ and $\bold{c}$ represent messages and codewords, respectively. A set $\Se{C}$ is called an $(n,k,d)_q$-code if $\Se{C}\subset \textup{GF}(q)^n$, $\Tx{dim}(\Se{C})=k$, and $\min\limits_{\bold{c}_1,\bold{c}_2\in \Se{C}, \bold{c}_1\neq\bold{c}_2} d_\textup{H}(\bold{c}_1,\bold{c}_2)=d$, where $d_\textup{H}$ refers to the Hamming distance. We next define a family of codes with double-level access. Note that our discussion is restricted to linear block codes.

\begin{defi}\label{defi: codeDL} Let $p,q\in\Db{N}$. Let $\bold{n}=(n_1,n_2,\dots,n_p)\in\Db{N}^p$, $\bold{k}=(k_1,k_2,\dots,k_p)\in\Db{N}^p$, $\bold{D}\in\Db{N}^{2\times p}$, $(\bold{D})_{x,y}=d_{x,y}$, where $d_{1,x}<d_{2,x}$, $k_{x}<n_{x}$, for all $x,y\in\MB{p}$. 

Let $n=n_1+n_2+\cdots+n_p$. Let $s_0=0$ and $s_{x}=n_1+n_2+\cdots+n_{x}$, $x\in\MB{p}$. Let $\bold{c}_{x}$ denote $\bold{c}\MB{s_{{x}-1}+1:s_{x}}$ and let $\bold{m}_x$ denote the message corresponding to $\bold{c}_x$, for $x\in\MB{p}$. A set $\Se{C}\subset \textup{GF}(q)^n$ is called an $(\bold{n},\bold{k},\bold{D},p)_q$-code if the following conditions are satisfied:
\begin{enumerate}
\item Let $\Se{C}_{x}=\LB{\bold{c}\MB{s_{x-1}+1:s_{x}}:\bold{c}\in\Se{C}}$, $x\in\MB{p}$. Each $\Se{C}_{x}$ is an $(n_{x},k_{x},d_{1,{x}})_q$-code.
\item Let $\Se{A}_{x}=\LB{\bold{c}\MB{s_{{x}-1}+1:s_{x}}:\bold{c}\in\Se{C}, \bold{c}\MB{s_{y-1}+1:s_{y}}=\bold{0}_{n_{y}}, \forall y\in\MB{p}\setminus \LB{x}}$, $x\in\MB{p}$. Each $\Se{A}_{x}$ is an $(n_{x},k_{x},d_{2,x})_q$-code.
\end{enumerate}
\end{defi}

\begin{exam}\label{exam: defi1} Let $q=16$ and $p=2$. Let $\bold{n}=(10,11)$ and $\bold{k}=(6,7)$. Then, $\bold{r}=\bold{n}-\bold{k}=(4,4)$. Suppose $\bold{D}$ is specified as follows:
\begin{equation}
\bold{D}=\left[\begin{array}{c|c}
4 & 3\\
\hline
7 & 6
\end{array}\right].
\end{equation}
Then, one can construct an $(\bold{n},\bold{k},\bold{D},p)_q$-code with the parameters specified previously.
\end{exam}

Any $(\bold{n},\bold{k},\bold{D},p)_q$-code specified according to \Cref{defi: codeDL} corrects $(d_{1,x}-1)$ erasures in the $i$-th local codeword via local access, and corrects additional $(d_{2,x}-d_{1,x})$ erasures through global access when other local codewords are all correctable via local access. Following this notation, \Cref{defi: codeHL} extends \Cref{defi: codeDL} into the triple-level case.

\begin{defi}\label{defi: codeHL} Let $q,p_0\in\Db{N}$, $\bold{p}=(p_1,p_2,\dots,p_{p_0})\in\Db{N}^{p_0}$, $p=p_1+p_2+\cdots+p_{p_0}$. Let $\bold{n}=(\bold{n}_1,\bold{n}_2,\dots,\bold{n}_{p_0})\in\Db{N}^{p_0}$, $\bold{k}=(\bold{k}_1,\bold{k}_2,\dots,\bold{k}_{p_0})\in\Db{N}^{p_0}$, where $\bold{n}_{x}=(n_{x,1},n_{x,2},\dots,n_{x,p_{x}})\in\Db{N}^{p_x}$, $\bold{k}_{x}=(k_{x,1},k_{x,2},\dots,k_{x,p_{x}})\in\Db{N}^{p_{x}}$, for all $x\in\MB{p_0}$. 

Let $t_0=0$, $t_{x}=p_1+p_2+\cdots+p_{x}$, $x\in\MB{p_0}$. Suppose $\bold{D}\in\Db{N}^{3\times p}$. Let $d_{l,x,i}=(\bold{D})_{l,t_{x-1}+i}$, $l\in\MB{3}$ so that $d_{1,x,i}<d_{2,x,i}<d_{3,x,i}$, for $x\in\MB{p_0}$ and $i\in\MB{p_{x}}$. Let $\bold{D}_{x}=\bold{D}\MB{1:2,t_{x-1}+1:t_{x}}$, $x\in\MB{p_0}$. Let $n_{x}=n_{x,1}+n_{x,2}+\cdots+n_{x,p_{x}}$ for all $x\in\MB{p_0}$. Let $n=n_1+n_2+\cdots+n_{p_0}$. Let $s_0=0$, $s_{x}=n_1+n_2+\cdots+n_{x}$, $x\in\MB{p_0}$. Let $s_{x,0}=s_{x}$, $s_{x,i}=s_{x}+n_{x,1}+n_{x,2}+\cdots+n_{x,i}$, for all $x\in\MB{p_0}$ and $i\in\MB{p_{x}}$. Let $\bold{c}_{x,i}$ denote $\bold{c}\MB{s_{x,i-1}+1:s_{x,i}}$ and let $\bold{m}_{x,i}$ denote the message corresponding to $\bold{c}_{x,i}$, for $x\in\MB{p_0}$, $i\in\MB{p_{x}}$. A set $\Se{C}\subset \textup{GF}(q)^n$ is called an $(\bold{n},\bold{k},\bold{D},p_0,\bold{p})_q$-code if the following conditions are satisfied:

\begin{enumerate}
\item Let $\Se{C}_{x}=\LB{\bold{c}\MB{s_{x-1}+1:s_{x}}:\bold{c}\in\Se{C}}$, $x\in\MB{p_0}$. Each $\Se{C}_{x}$ is an $(\bold{n}_{x},\bold{k}_{x},\bold{D}_{x},p_{x})_q$-code.
\item Let $\Se{A}_{x,i}=\{\bold{c}\MB{s_{x,i-1}+1:s_{x,i}}:\bold{c}\in\Se{C},\bold{c}\MB{s_{y,j-1}+1:s_{y,j}}=\bold{0}_{n_{y,j}}, \forall y\in\MB{p_0},j\in\MB{p_{y}},(x,i)\neq (y,j)\}$. Each $\Se{A}_{x}$ is an $(n_{x,i},k_{x,i},d_{3,x,i})_q$-code. 
\end{enumerate}
\end{defi}

\begin{exam}\label{exam: defi2} Let $q=16$, $p_0=3$, and $\bold{p}=(2,2,4)$. Let $\bold{n}=(\bold{n}_1,\bold{n}_2,\bold{n}_3)$, where $\bold{n}_1=(10,11)$, $\bold{n}_2=(10,10)$, and $\bold{n}_3=(12,12,12,12)$. Let $\bold{k}=(\bold{k}_1,\bold{k}_2,\bold{k}_3)$, where $\bold{k}_1=(6,6)$, $\bold{k}_2=(7,7)$, and $\bold{k}_3=(9,8,9,9)$. Then, $\bold{r}=\bold{n}-\bold{k}=(\bold{r}_1,\bold{r}_2,\bold{r}_3)$, where $\bold{r}_1=(4,5)$, $\bold{r}_2=(3,3)$, $\bold{r}_3=(3,4,3,3)$. Suppose $\bold{D}$ is specified as follows:
\begin{equation}
\bold{D}=\left[\begin{array}{cc|cc|cccc}
2 &3 & 2 & 2 & 2 & 2 & 2 & 2\\ 
\hline
6 &7 & 5 & 5 & 8 & 8 & 8 & 8\\ 
\hline
9 &10 & 9 & 9 & 11 & 11 & 11 & 11 
\end{array}\right].
\end{equation}
Then, one can construct an $(\bold{n},\bold{k},\bold{D},p_0,\bold{p})_q$-code with the parameters specified previously.
\end{exam}

This definition can be easily generalized into codes with more than three levels of access. For simplicity, we constrain our discussion to the triple-level case.

\subsection{Cauchy Matrices}
\label{subsection: CauchyMatrices}
Cauchy matrices are the key component in the construction that we will introduce shortly. 

\begin{defi} (Cauchy matrix) \label{CauchyMatrix} Let $s,t\in\Db{N}$ and $\textup{GF}(q)$ be a finite field of size $q$. Suppose $a_1,\dots,a_x,b_1,\dots,b_y$ are pairwise distinct elements in $\textup{GF}(q)$. The following matrix is known as a \textbf{Cauchy matrix},

\begin{equation*}\left[
\begin{array}{cccc}
\frac{1}{a_1-b_1} & \frac{1}{a_1-b_2} & \dots & \frac{1}{a_1-b_t}\\
\frac{1}{a_2-b_1} & \frac{1}{a_2-b_2} & \dots & \frac{1}{a_2-b_t}\\
\vdots & \vdots &\ddots & \vdots \\
\frac{1}{a_s-b_1} & \frac{1}{a_s-b_2} & \dots & \frac{1}{a_s-b_t}\\
\end{array}\right].
\label{defi: Cauchy matrix}
\end{equation*}
We denote this matrix by $\bold{Y}(a_1,\dots,a_s;b_1,\dots,b_t)$.
\end{defi}

Cauchy matrices are \emph{totally invertible}, i.e., every square sub-matrix of a Cauchy matrix is invertible. The inverse of a given Cauchy matrix can be explicitly computed using algorithms of lower complexity than those for inverting Vandermonde matrices. These properties make Cauchy matrices promising in designing systematic maximum distance separable (MDS) codes. \Cref{lemma: Good matrix} presents a useful result about Cauchy matrices that will be used repeatedly in this paper.

\begin{lemma} \label{lemma: Good matrix} Let $s,t,r\in\Db{N}$ such that $t-s<r\leq t$, $\bold{A}\in \textup{GF}(q)^{s\times t}$. If $\bold{A}$ is a Cauchy matrix, then the following matrix $\bold{M}$ is a parity-check matrix of an $(s+r,s+r-t,t+1)_q$-code.
\begin{equation*}
\bold{M}=\left[
\begin{array}{c}
\bold{A}\\
-\bold{I}_r\ \bold{0}_{r\times(t-r)}\\
\end{array}
\right]^{\Tx{T}}.
\end{equation*}
\end{lemma} 

\begin{proof} The parity-check matrix of an $(s+r,s+r-t,t+1)_q$-code satisfies the property that every $t$ columns of this matrix are linearly independent. Therefore, we only need to prove that every $t$ rows of $\bold{M}^{\Tx{T}}$ are linearly independent. We prove Lemma~\ref{lemma: Good matrix} by contradiction. Suppose there exist $t$ rows from $\bold{M}^{\Tx{T}}$ that are linearly dependent. Suppose $a$ of these linearly dependent rows $\bold{r}_1,\bold{r}_2,\dots,\bold{r}_a$ are from $\bold{A}$, and the other $t-a$ rows $\bold{r}_{a+1},\bold{r}_{a+2},\dots,\bold{r}_{t}$ are from $\MB{-\bold{I}_r\ \bold{0}_{r\times(t-r)}}$, where $0\leq t-a\leq r$. Suppose the entries with $-1$ in $\bold{r}_{a+1},\bold{r}_{a+2},\dots,\bold{r}_{t}$ are located in the $i_1,i_2\dots,i_{t-a}$-th columns of $\bold{M}^{\textup{T}}$, then $i_p\leq r$ for all $1\leq p\leq t-a$. Observe that $\MB{t}$ is the set of indices of all columns in $\bold{M}^{\Tx{T}}$. Suppose $\MB{t}\setminus \LB{i_1,i_2,\dots,i_{t-a}}=\LB{j_1,j_2,\dots,j_a}$. Then the $a\times a$ sub-matrix of the intersection of the rows $\bold{r}_1,\bold{r}_{2},\dots,\bold{r}_a$ and the $j_1,j_2,\dots,j_a$-th columns of $\bold{A}$ is singular. A contradiction.
\end{proof}

\section{Codes for Multi-Level Access}
\label{section: codes for multi-level access}
Following the definitions and notation introduced in \Cref{section: notation and preliminaries}, we present a CRS-based code with double-level access in \Cref{subsection: DLaccessCodes}. Then, we extend our construction into a triple-level case in \Cref{subsection: codes with hierarchical locality}.

\subsection{Codes with Double-Level Access}
\label{subsection: DLaccessCodes}

In this subsection, we provide a construction of codes offering double-level access based on the CRS codes. Note that the generator matrix of any systematic code with double-level access has the following structure:
\begin{equation}\label{eqn: GenMatDL}\bold{G}=\left[
\begin{array}{c|c|c|c|c|c|c}
\bold{I}_{k_1} & \bold{A}_{1,1} & \bold{0} & \bold{A}_{1,2} & \dots & \bold{0} & \bold{A}_{1,p}\\
\hline
\bold{0} & \bold{A}_{2,1} & \bold{I}_{k_2} & \bold{A}_{2,2} & \dots & \bold{0}& \bold{A}_{2,p}\\
\hline
\vdots & \vdots & \vdots & \vdots & \ddots & \vdots & \vdots \\
\hline
\bold{0} & \bold{A}_{p,1} & \bold{0} & \bold{A}_{p,2} & \dots & \bold{I}_{k_p} & \bold{A}_{p,p}\\
\end{array}\right].
\end{equation}

\begin{cons}\label{cons: CRScons} (CRS-based code) Let $p\in\Db{N}$, $k_1,k_2,\dots,k_p\in \Db{N}$, $n_1,n_2,\dots,n_p\in\Db{N}$, $\delta_1,\delta_2,\dots,\delta_p\in \Db{N}$ and $\delta=\delta_1+\delta_2+\dots+\delta_p$, with $r_x=n_x-k_x>0$ for all $x\in\MB{p}$. Let $GF(q)$ be a finite field such that $q\geq \max\nolimits_{x\in\MB{p}}\LB{n_x}+\delta$. 

For each $x\in \MB{p}$, let $a_{x,i}$, $b_{x,j}$, $i\in\MB{k_x+\delta_x}$, $j\in\MB{r_x-\delta_x+\delta}$, be distinct elements of $\textup{GF}(q)$.
Consider the Cauchy matrix $\bold{T}_x\in \textup{GF}(q)^{(k_x+\delta_x)\times (r_x-\delta_x+\delta)}$ such that $\bold{T}_x=\bold{Y}(a_{x,1}, \dots, a_{x,k_x+\delta_x};b_{x,1},\dots,b_{x,r_x-\delta_x+\delta})$. For each $x\in\MB{p}$, we obtain $\{\bold{B}_{x,i}\}_{i\in\MB{p}\setminus\{x\}}$, $\bold{U}_x$, $\bold{A}_{x,x}$, according to the following partition of $\bold{T}_x$,

\begin{equation}\label{eqn: CRS}
\bold{T}_x=\left[
\begin{array}{c|c}
\bold{A}_{x,x} & \begin{array}{c|c|c}
\bold{B}_{x,1} & \dots & \bold{B}_{x,p}
\end{array}
\\
\hline
\bold{U}_x & \bold{Z}_{x}
\end{array}\right],
\end{equation}
where $\bold{A}_{x,x}\in \textup{GF}(q)^{k_x\times r_x}$, $\bold{B}_{x,i}\in \textup{GF}(q)^{k_x\times \delta_i}$, $\bold{U}_x\in \textup{GF}(q)^{\delta_x\times r_x}$. Moreover, $\bold{A}_{x,y}=\bold{B}_{x,y}\bold{U}_y$, for $x \neq y$.

Matrices $\bold{A}_{x,x}$ and $\bold{A}_{x,y}$ are substituted in $\bold{G}$ specified in (\ref{eqn: GenMatDL}), for all $x,y\in\MB{p}$, $x\neq y$. Let $\Se{C}_1$ represent the code with generator matrix $\bold{G}$.
\end{cons}

\begin{lemma}\label{lemma: DLcodedis} Following the notation in \Cref{defi: codeDL}, let $d_{1,x}=r_x-\delta_x+1$, $d_{2,x}=r_x-\delta_x+\delta+1$, for $x\in\MB{p}$. Then, code $\Se{C}_1$ specified in \Cref{cons: CRScons} is an $(\bold{n},\bold{k},\bold{D},p)_q$-code.
\end{lemma}

\begin{proof}[Sketch of the proof] For each $x\in \MB{p}$, define $\bold{y}_x=\sum\nolimits_{y\in\MB{p},y\neq x} \bold{m}_y\bold{B}_{y,x}$. It follows from $\bold{m}\bold{G}=\bold{c}$ and (\ref{eqn: GenMatDL}) that for $x\in\MB{p}$, $\bold{c}_x=\MB{\bold{m}_x,\bold{m}_x\bold{A}_{x,x}+\bold{y}_x\bold{U}_x}$. Define the local parity-check matrix $\bold{H}^{\Tx{L}}_x$ and the global parity-check matrix $\bold{H}^{\Tx{G}}_x$, for each $x\in\MB{p}$, as follows:
\begin{equation*}
\bold{H}_x^{\Tx{G}}=\left[
\begin{array}{c|c}
\bold{A}_{x,x} & \begin{array}{c|c|c}
\bold{B}_{x,1} & \dots & \bold{B}_{x,p}
\end{array}
\\
\hline
-\bold{I}_{r_x} & \bold{0}_{r_x \times \delta-\delta_x}
\end{array}\right]^{\Tx{T}}, \bold{H}^{\Tx{L}}_x=\left[\begin{array}{ccc}
\bold{A}_{x,x}\\
-\bold{I}_{r_x}\\
\bold{U}_x\\
\end{array}\right]^{\Tx{T}}.
\end{equation*}
We next prove the equations of the local distance $d_{1,x}=r_x-\delta_x+1$ and the global distance $d_{2,x}=r_x-\delta_x+\delta+1$ using $\bold{H}^{\Tx{L}}_x$ and $\bold{H}^{\Tx{G}}_x$, $x\in\MB{p}$.

To prove the equation of the local distance, let $\tilde{\bold{c}}_x=\MB{\bold{c}_x,\bold{y}_x}$. Then, one can show that $\tilde{\bold{c}}_x$ belongs to a code $\Se{C}_x^{\Tx{L}}$ with the local parity-check matrix $\bold{H}^{\Tx{L}}_x$. From \Cref{lemma: Good matrix}, $\Se{C}_x^{\Tx{L}}$ is an $(n_x+\delta_x,k_x,r_x+1)_q$-code. Therefore, any $r_x$ erasures in $\tilde{\bold{c}}_x$ are correctable. Provided that $\bold{y}_x$ has length $\delta_x$, we can consider the entries of $\bold{y}_x$ as erasures and thus any $(r_x-\delta_x)$ erasures in the remaining part of $\tilde{\bold{c}}_x$, i.e., $\bold{c}_x$, can be corrected. Therefore, $d_{1,x}=r_x-\delta_x+1$.

To prove the equation of the global distance, assume all the local codewords except for $\bold{c}_x$ are successfully decodable locally. Then, for each $x\in\MB{p}$, $\bold{y}_x$ and $\bold{s}_{x}=\MB{\bold{m}_x\bold{B}_{x,1},\dots,\bold{m}_x\bold{B}_{x,p}}$ are computable. Let $\bar{\bold{c}}_x=\bold{c}_x-\MB{\bold{0}_{k_x},\bold{y}_x\bold{U}_x}$, then one can show that $\bold{H}^{\Tx{G}}_x\bar{\bold{c}}_x^{\Tx{T}}=\MB{\bold{0}_{r_x}, \bold{s}_{x}}^{\Tx{T}}$. From \Cref{lemma: Good matrix} and from the construction of $\bold{H}^{\Tx{G}}_x$, any $(r_x-\delta_x+\delta)$ erasures in $\bar{\bold{c}}_x$ are correctable, thus $(r_x-\delta_x+\delta)$ erasures in $\bold{c}_x$ are also correctable. Therefore, $d_{2,x}=r_x-\delta_x+\delta+1$.
\end{proof}

We next provide a working example for codes in \Cref{cons: CRScons}. For simplicity, we let all the local codeword lengths and local data lengths be equal. However, the construction itself allows them to be unequal.

\begin{table}
\centering
\caption{Polynomial and normal forms of $\textup{GF}(2^4)$}
\begin{tabular}{|c|c||c|c||c|c||c|c|}
\hline
$0$ & $0000$ & $\beta^4$ & $1100$ & $\beta^{8}$ & $1010$ & $\beta^{12}$ & $1111$\\
\hline
$\beta$ & $0100$ & $\beta^{5}$ & $0110$ & $\beta^{9}$ & $0101$ & $\beta^{13}$ & $1011$\\
\hline
$\beta^2$ & $0010$ & $\beta^{6}$ & $0011$ & $\beta^{10}$ & $1110$ & $\beta^{14}$ & $1001$\\
\hline
$\beta^3$ & $0001$ & $\beta^{7}$ & $1101$ & $\beta^{11}$ & $0111$ & $\beta^{15}=1$ & $1000$\\
\hline
\end{tabular}
\label{table: GF}
\end{table}

\begin{figure*}
\normalsize
\setcounter{equation}{4}
\begin{equation}\label{eqn: exam1}\small
\bold{T}_{1}=\bold{T}_{2}=\left[\begin{array}{c|c}
\bold{A}_{1,1} & \bold{B}_{1,2}\\
\hline
\bold{U}_1 & \bold{Z}_1
\end{array}\right]=\left[\begin{array}{c|c}
\bold{A}_{2,2} & \bold{B}_{2,1}\\
\hline
\bold{U}_2 & \bold{Z}_2
\end{array}\right]
=\left[\begin{array}{ccc|c}
\frac{1}{\beta^{}-\beta^{8}} & \frac{1}{\beta^{}-\beta^{9}} & \frac{1}{\beta^{}-\beta^{10}} & \frac{1}{\beta-\beta^{11}}\\
\frac{1}{\beta^{2}-\beta^{8}} & \frac{1}{\beta^{2}-\beta^{9}} & \frac{1}{\beta^{2}-\beta^{10}} &\frac{1}{\beta^{2}-\beta^{11}}\\
\frac{1}{\beta^{3}-\beta^{8}} & \frac{1}{\beta^{3}-\beta^{9}} & \frac{1}{\beta^{3}-\beta^{10}} &\frac{1}{\beta^{3}-\beta^{11}} \\
\hline
\frac{1}{\beta^{7}-\beta^{8}} & \frac{1}{\beta^{7}-\beta^{9}} & \frac{1}{\beta^{7}-\beta^{10}} & \frac{1}{\beta^{7}-\beta^{11}}
\end{array}\right]=\left[\begin{array}{ccc|c}
\beta^{5} &\beta^{12} & \beta^{7} & \beta^{9}\\
1 &\beta^{4} & \beta^{11} & \beta^{6}\\
\beta^{2} &\beta^{14} & \beta^{3} & \beta^{10}\\
\hline
\beta^{4} & 1 & \beta^{9} & \beta^{7} 
\end{array}
\right].
\end{equation}
\hrulefill
\setcounter{equation}{5}
\end{figure*}

\begin{exam} \label{exam: CodeDL} Let $q=2^4$, $p=2$, $r=r_1=r_2=3$, $\delta'=\delta_1=\delta_2=1$, $k=k_1=k_2=3$, $n=n_1=n_2=k+r=6$, $\delta=\delta_1+\delta_2=2$. Then, $d_1=r-\delta'+1=3-1+1=3$, $d_2=r-\delta'+\delta+1=3-1+2+1=5$. Choose a primitive polynomial over $\textup{GF}(2)$: $g(X)=X^4+X+1$. Let $\beta$ be a root of $g(X)$, then $\beta$ is a primitive element of $\textup{GF}(2^4)$. The binary representation of all the symbols in $\textup{GF}(2^4)$ is specified in \Cref{table: GF}. 

Let $\bold{A}_{1,1}=\bold{A}_{2,2}$, $\bold{B}_{1,2}=\bold{B}_{2,1}$, $\bold{U}_1=\bold{U}_2$, and $\bold{T}_1=\bold{T}_2$ as specified in (\ref{eqn: exam1}).
Therefore, 
\begin{equation*}\bold{A}_{1,2}=\bold{A}_{2,1}=\bold{B}_{2,1}\bold{U}_1=\left[\begin{array}{ccc}
\beta^{13} &\beta^{9} & \beta^{3}\\
\beta^{10} &\beta^{6} & 1\\
\beta^{14} &\beta^{10} & \beta^{4}
\end{array}
\right].
\end{equation*}

Then, the generator matrix $\bold{G}$ is specified as follows,
\begin{equation*}\small
\left[
\begin{array}{ccc|ccc|ccc|ccc}
1 & 0 & 0 & \beta^{5} &\beta^{12} & \beta^{7} & 0 & 0 & 0 & \beta^{13} &\beta^{9} & \beta^{3}\\
0 & 1 & 0 & 1 &\beta^{4} & \beta^{11} & 0 & 0 & 0 & \beta^{10} &\beta^{6} & 1 \\
0 & 0 & 1 & \beta^{2} &\beta^{14} & \beta^{3} & 0 & 0 & 0 & \beta^{14} &\beta^{10} & \beta^{4} \\
\hline
0 & 0 & 0 & \beta^{13} &\beta^{9} & \beta^{3} & 1 & 0 & 0 & \beta^{5} &\beta^{12} & \beta^{7} \\
0 & 0 & 0 &\beta^{10} &\beta^{6} & 1 & 0 & 1 & 0 & 1 &\beta^{4} & \beta^{11} \\
0 & 0 & 0 & \beta^{14} &\beta^{10} & \beta^{4} & 0 & 0 & 1 & \beta^{2} &\beta^{14} & \beta^{3} \\
\end{array}\right].
\end{equation*}
Suppose $\bold{m}_1=(1,\beta,\beta^2)$, $\bold{m}_2=(\beta,1,0)$, then $\bold{c}_1=(1,\beta,\beta^2,\beta^{14},0,0)$ and $\bold{c}_2=(\beta,1,0,\beta^{6},0,\beta^{13})$. Moreover, $\bold{H}_1^{\Tx{L}}$ and $\bold{H}_1^{\Tx{G}}$ are specified as follows,
\begin{equation*}\small
\bold{H}_1^{\Tx{G}}=\left[\hspace{-0.1cm}\begin{array}{cccc}
\beta^5 & \beta^{12} & \beta^7 & \beta^9\\
1 & \beta^{4} & \beta^{11} & \beta^6 \\
\beta^2 & \beta^{14} & \beta^3 &\beta^{10}\\
1 & 0 & 0 & 0 \\
0 & 1 & 0 & 0 \\
0 & 0 & 1 & 0
\end{array}\hspace{-0.1cm}\right]^{\Tx{T}},\bold{H}_1^{\Tx{L}}=\left[\hspace{-0.1cm}\begin{array}{cccc}
\beta^5 & \beta^{12} & \beta^7 \\
1 & \beta^{4} & \beta^{11}\\
\beta^2 & \beta^{14} & \beta^3\\
1 & 0 & 0 \\
0 & 1 & 0 \\
0 & 0 & 1 \\
\beta^4 & 1 & \beta^9
\end{array}\hspace{-0.1cm}\right]^{\Tx{T}}.
\end{equation*}
According to \Cref{cons: CRScons}, $\bold{G}$ is the generator matrix of a double-level accessible code that corrects $2$ local erasures by local access and corrects $2$ extra erasures within a single local cloud by global access. In the following, we denote the erased version of $\bold{c}_1$ by $\bold{c}'_1$, and erased symbols by $e_i$, $i\in\Db{N}$.

As an example of decoding by local access, suppose $\bold{c}'_1=(1,e_1,\beta^2,e_2,0,0)$. Then, the erased elements of $\tilde{\bold{c}}_1=(1,e_1,\beta^2,e_2,0,0,e_3)$ can be retrieved using $\bold{H}_1^{\Tx{L}}$ as the parity-check matrix. In particular, we solve $\bold{H}^{\Tx{L}}_1\tilde{\bold{c}}_1^{\Tx{T}}=(0,0,0)^{\Tx{T}}$ for $e_1,e_2,e_3$ and obtain $(e_1,e_2,e_3)=(\beta,\beta^{14},\beta^7)$. We have decoded $\bold{c}_1$ successfully.

As an example of decoding by global access, suppose $\bold{c}'_1=(e_1,e_2,\beta^2,e_3,e_4,0)$, and $\bold{c}_2$ has been locally decoded successfully. Then, $\bold{c}_2=(\beta,1,0,\beta^{6},0,\beta^{13})$ implies that $\bold{m}_1\bold{B}_{1,2}\bold{U}_2=(\beta^{6},0,\beta^{13})-\beta\cdot(\beta^5,\beta^{12},\beta^7)-1\cdot(1,\beta^4,\beta^{11})=(1,\beta^{11},\beta^{5})$. Since $\bold{U}_2=(\beta^{4},1,\beta^9)$, we obtain $\bold{m}_1\bold{B}_{1,2}=\beta^{11}$. Moreover, we compute $\bold{m}_2\bold{B}_{2,1}\bold{U}_1=(\beta^{11},\beta^{7},\beta)$. Let $\bar{\bold{c}}_1=\bold{c}'_1-(0,0,0,\beta^{11},\beta^{7},\beta)=(e'_1,e'_2,\beta^2,e'_3,e'_4,\beta)$. Then, we solve $\bold{H}^{\Tx{G}}_1\bar{\bold{c}}_1^{\Tx{T}}=(0,0,0,\beta^{11})^{\Tx{T}}$ and obtain $(e'_1,e'_2,e'_3,e'_4)=(1,\beta,\beta^{10},\beta^7)$. Therefore, $e_1=e'_1=1$, $e_2=e'_2=\beta$, $e_3=e'_3+\beta^{11}=\beta^{14}$, $e_4=e'_4+\beta^7=0$, and we have decoded $\bold{c}_1$ successfully.
\end{exam}

\subsection{Codes with Hierarchical Locality}
\label{subsection: codes with hierarchical locality}
Based on the double-level accessible codes presented in \Cref{subsection: DLaccessCodes}, we present a class of codes with hierarchical locality in \Cref{cons: ConsHL}. For simplicity, we just present a construction with triple-level access. Note that the coding scheme itself can be naturally extended to have more than three levels.

As described in \Cref{defi: codeHL}, in the triple-level structure, the set of local clouds is partitioned into $p_0$ groups that are indexed by the first-level index $x\in\MB{p_0}$. These groups are further divided into $p_1,p_2,\dots,p_{p_0}$ local clouds, respectively, and the local clouds within group $x$ are indexed by the second-level index $i\in\MB{p_x}$. Therefore, each local cloud is indexed by the pair $(x,i)$. In the following discussion, the parameters with subscript $(x,y;i,j)$ are determined via the two local clouds indexed by $(x,i)$ and $(y,j)$. The subscript $(x,y;i)$ is an abbreviated version of $(x,y;i,1),(x,y;i,2),\dots,(x,y;i,p_y)$, and the parameters with subscript $(x,y;i)$ are determined via the local cloud $(x,i)$ and all the local clouds in the $y$-th group. Lastly, we define a new notation, $(x,y;i;s)$, that indexes the parameters determined via the local cloud $(x,i)$ and some other local clouds in the $y$-th group (not necessarily all of them). Note that this notation bares similarity to $(x,y;i,j)$. However, they are different notations: the index $s$ indexes a subgroup of local clouds not a single one as done by $j$.

A generator matrix of such a code is as follows:
\begin{equation}\label{eqn: GenMatHL}
\bold{G}=\left[\begin{array}{c|c|c|c}
\bold{F}_{1,1} & \bold{F}_{1,2} & \dots & \bold{F}_{1,p_0}\\
\hline
\bold{F}_{2,1} & \bold{F}_{2,2} & \dots & \bold{F}_{2,p_0}\\
\hline
\vdots & \vdots & \ddots & \vdots \\
\hline
\bold{F}_{p_0,1} & \bold{F}_{p_0,2} & \dots & \bold{F}_{p_0,p_0}\\
\end{array}
\right], 
\end{equation}
where for any $x\in\MB{p_0}$,
\begin{equation}\label{eqn: MainMatHL}
\bold{F}_{x,x}=\left[\begin{array}{c|c|c|c|c}
\bold{I}_{k_{x,1}} & \bold{A}_{x,x;1,1} & \dots &\bold{0} & \bold{A}_{x,x;1,p_x}\\
\hline
\vdots & \ddots & \ddots & \vdots & \vdots\\
\hline
\bold{0} & \bold{A}_{x,x;p_x,1} & \dots &\bold{I}_{k_{x,p_x}} & \bold{A}_{x,x;p_x,p_x}\\
\end{array}
\right],
\end{equation}
is a generator matrix of a code offering double-level access, and 
\begin{equation}\label{eqn: CrossMatHL}
\bold{F}_{x,y}=\left[\begin{array}{c|c|c|c|c}
\bold{0} & \bold{A}_{x,y;1,1} & \dots &\bold{0} & \bold{A}_{x,y;1,p_y}\\
\hline
\vdots & \ddots & \ddots & \vdots & \vdots\\
\hline
\bold{0} & \bold{A}_{x,y;p_x,1} & \dots &\bold{0} & \bold{A}_{x,y;p_x,p_y}\\
\end{array}
\right].
\end{equation}
Properties of $\bold{F}_{x,x},\bold{F}_{x,y}$ are to be discussed later.

\begin{cons}\label{cons: ConsHL} Let $p_0\in\Db{N}$, $\bold{p} = (p_1,\dots,p_{p_0}) \in\Db{N}^{p_0}$. Let $k_{x,i}, n_{x,i}, \delta_{x,i}, \gamma_{x}\in \Db{N}$, for $x\in \MB{p_0}$ and $i\in \MB{p_x}$, such that $r_{x,i}=n_{x,i}-k_{x,i}>0$ and $2\gamma_x<\min\nolimits_{i\in \MB{p_x}}\{r_{x,i}-\delta_{x,i}\}$. Let $\delta_x=\delta_{x,1}+\cdots+\delta_{x,p_x}$, $\gamma=\sum\nolimits_{x\in \MB{p_0}} p_x\gamma_x$, for all $x\in\MB{p_0}$. Let $GF(q)$ be a finite field such that $q\geq \max\limits_{x\in\MB{p_0},i\in\MB{p_x}}\{n_{x,i}+\delta_x-(p_x-2)\gamma_x+\gamma\}$.

Let $u_{x,i}=k_{x,i}+\delta_{x,i}+2\gamma_x$, $v_{x,i}=r_{x,i}-\delta_{x,i}+\delta_x-p_x\gamma_x+\gamma$, for $x\in\MB{p_0}$, $i\in\MB{p_x}$. For each $x\in\MB{p_0}$, $i\in \MB{p_x}$, let $a_{x,i,s}, b_{x,i,t}$, $s\in \MB{u_{x,i}}$, $t\in\MB{v_{x,i}}$, be distinct elements of $\textup{GF}(q)$.

Consider the Cauchy matrix $\bold{T}_{x,i}$ on $\textup{GF}(q)^{u_{x,i}\times v_{x,i}}$ such that $\bold{T}_{x,i}=\bold{Y}(a_{x,i,1},\dots,a_{x,i,u_{x,i}}; b_{x,i,1},\dots,b_{x,i,v_{x,i}})$, for $x\in\MB{p_0}$, $i\in\MB{p_x}$. Then, we obtain $\bold{A}_{x,x;i,i}$, $\bold{B}_{x,x;i,i'}$, $\bold{E}_{x,y;i;j}$, $\bold{U}_{x,i}$, $\bold{V}_{x,i}$, $x\in\MB{p_0}$, $i'\in\MB{p_x}\setminus\LB{i}$, $y\in \MB{p_0}\setminus \LB{x}$, $j\in\MB{p_y}$, according to the following partition of $\bold{T}_{x,i}$,

\begin{equation}\label{eqn: CRSHL}
\bold{T}_{x,i}=\left[
\begin{array}{c|c}
\bold{A}_{x,x;i,i} & \begin{array}{c|c|c|c}
\bold{B}_{x,x;i} & \bold{E}_{x,1;i} & \dots & \bold{E}_{x,p_0;i}
\end{array}
\\
\hline
\begin{array}{c}\bold{U}_{x,i}\\
\hline 
\bold{V}_{x,i}\end{array}& \bold{Z}_{x,i}
\end{array}\right],
\end{equation} 
\begin{equation}
\textit{where } \text{ } \bold{B}_{x,x;i}=\left[\begin{array}{c|c|c}\bold{B}_{x,x;i,1} & \dots & \bold{B}_{x,x;i,p_x}
\end{array}
\right]
\end{equation}
\begin{equation}
\textit{and  } \text{ } \bold{E}_{x,y;i}=\left[\begin{array}{c|c|c}\bold{E}_{x,y;i;1} & \dots & \bold{E}_{x,y;i;p_y}
\end{array}
\right], 
\end{equation}
such that $\bold{A}_{x,x;i,i}\in \textup{GF}(q)^{k_{x,i}\times r_{x,i}}$, $\bold{B}_{x,x;i,i'}\in \textup{GF}(q)^{k_{x,i}\times \delta_{x,i'}}$, $\bold{E}_{x,y;i;j}\in \textup{GF}(q)^{k_{x,i}\times \gamma_y}$, $\bold{U}_{x,i}\in \textup{GF}(q)^{\delta_{x,i}\times r_{x,i}}$, $\bold{V}_{x,i}\in \textup{GF}(q)^{2\gamma_x \times r_{x,i}}$. Moreover,  $\bold{A}_{x,x;i,i'}=\bold{B}_{x,x;i,i'}\bold{U}_{x,i'}$. Suppose $\bold{E}_{x,y;i;p_y+1}=\bold{E}_{x,y;i;1}$; let $\bold{A}_{x,y;i,j}=\MB{\bold{E}_{x,y;i;j},\bold{E}_{x,y;i;j+1}}\bold{V}_{y,j}$.

Matrices $\bold{A}_{x,x; i,i}$ and $\bold{A}_{x,y; i,j}$ are substituted in $\bold{F}_{x,x}$ and $\bold{F}_{x,y}$ to construct $\bold{G}$ as specified in (\ref{eqn: GenMatHL}), (\ref{eqn: MainMatHL}), and (\ref{eqn: CrossMatHL}). Let $\Se{C}_2$ represent the code with generator matrix $\bold{G}$.

\end{cons}

\begin{theo}\label{theo: HLdis} Following the notation in \Cref{defi: codeHL}, let $d_{1,x,i}=r_{x,i}-\delta_{x,i}-2\gamma_x+1$, $d_{2,x,i}=r_{x,i}-\delta_{x,i}+\delta_x+1$, $d_{3,x,i}=r_{x,i}-\delta_{x,i}+\delta_x-p_x\gamma_x+\gamma+1$, for $x\in\MB{p_0}$, $i\in\MB{p_x}$. Then, the code $\Se{C}_2$ defined in \Cref{cons: ConsHL} is an $(\bold{n},\bold{k},\bold{D},p_0,\bold{p})_q$-code. 
\end{theo}

\begin{proof}[Sketch of the proof] For each $x\in\MB{p_0}$ and $i\in\MB{p_x}$, define the local cross parity $\bold{y}_{x,i}=\sum\nolimits_{i'\in\MB{p_x}\setminus\{i\}}\bold{m}_{x,i'}\bold{B}_{x,x;i,i'}$, and the global cross parities $\bold{\bold{z}}_{x,i}=\sum\nolimits_{y\in\MB{p_0}\setminus \LB{x}, j\in \MB{p_y}}\bold{m}_{y,j}\bold{E}_{y,x;j;i}$. Let $\bold{z}_{x,p_x+1}=\bold{z}_{x,p_x}$. Then, it follows from $\bold{m}\bold{G}=\bold{c}$ that $\bold{c}_{x,i}=\MB{\bold{m}_{x,i},\bold{w}_{x,i}}$ for some $\bold{w}_{x,i}=\bold{m}_{x,i}\bold{A}_{x,x;i,i}+\bold{y}_{x,i}\bold{U}_{x,i}+\MB{\bold{z}_{x,i},\bold{z}_{x,i+1}}\bold{V}_{x,i}$.

The local erasure-correction capability $d_{1,x,i}=r_{x,i}-\delta_{x,i}-2\gamma_x+1$ and the global erasure-correction capability $d_{3,x,i}=r_{x,i}-\delta_{x,i}+\delta_x-p_x\gamma_x+\gamma+1$ can be easily derived by following the same logic used in the proof of \Cref{lemma: DLcodedis}. Therefore, we only need to prove that $d_{2,x,i}=r_{x,i}-\delta_{x,i}+\delta_x+1$.

To prove this statement, suppose all the local codewords in the $x$-th group except for $\bold{c}_{x,i}$ are successfully decodable locally, for some $x\in\MB{p_0}$, $i\in \MB{p_x}$. In other words, for all $i'\in\MB{p_x}\setminus\{i\}$, there are at most $d_{1,x,i'}-1$ erasures in the corrupted version $\bold{c}_{x,i'}$ of the local codeword. From the construction, we know that the row spaces of any two matrices from $\bold{A}_{x,x;i,i}$, $\bold{U}_{x,i}$, and $\bold{V}_{x,i}$ have no common elements except for the all zero vector. Therefore, for all $i'\in\MB{p_x}\setminus\{i\}$, $\bold{m}_{x,i'}$, $\bold{y}_{x,i'}$, $\MB{\bold{z}_{x,i'},\bold{z}_{x,i'+1}}$, can all be derived from $\bold{c}_{x,i}$. This implies that $\MB{\bold{z}_{x,i},\bold{z}_{x,i+1}}$ is known and thus, the entire contribution of global cross parities can be removed. Namely, let $\tilde{\bold{c}}_{x,i'}=\bold{c}_{x,i'}-\MB{\bold{0}_{k_{x,i'}},\MB{\bold{z}_{x,i'},\bold{z}_{x,i'+1}}\bold{V}_{x,i'}}$, for all $i'\in\MB{p_x}$, then the message $\bold{m}_{x}\bold{F}_{x,x}=\tilde{\bold{c}}_{x}$, where $\tilde{\bold{c}}_x=\MB{\tilde{\bold{c}}_{x,1},\dots,\tilde{\bold{c}}_{x,p_x}}$. Thus, from \Cref{lemma: DLcodedis}, $(r_{x,i}-\delta_{x,i}+\delta_x)$ erasures in $\tilde{\bold{c}}_{x,i}$ are correctable. Therefore, $d_{2,x,i}=r_{x,i}-\delta_{x,i}+\delta_x+1$.
\end{proof}

\begin{rem}\label{rem:1} Note that the constraint of $\gamma_y\in\Db{N}$ in \Cref{cons: CRScons} can be relaxed to $2\gamma_y\in\Db{N}$ if $p_y$ is even. In this case, we have $\bold{E}_{x,y;i;j}\in \textup{GF}(q)^{k_{x,i}\times 2\gamma_y}$. Moreover, we need to modify the equation of $\bold{E}_{x,y;i}$ to be $\bold{E}_{x,y;i}=\MB{\bold{E}_{x,y;i;1},\dots,\bold{E}_{x,y;i;p_y/2}}$, and $\bold{A}_{x,y;i,j}=\bold{E}_{x,y;i;\lceil j/2 \rceil}\bold{V}_{y,j}$.
\end{rem}

The following is a working example of \Cref{cons: ConsHL}. For simplicity, we let the middle code be the code presented in \Cref{exam: CodeDL}. However, the construction itself doesn't impose any constraints on $r_{x,i}$, $\delta_{x,i}$, and $\gamma_{x}$, except for $2\gamma_x<\min\nolimits_{y\in \MB{p_x}}\{r_{x,y}-\delta_{x,y}\}$.

\begin{exam}\label{exam: CodeHL} Here, we build on \Cref{exam: CodeDL} using the same $GF(q)$. Let $p_0=2$, $\bold{p}=(p_1,p_2)=(2,2)$, $\gamma'=\gamma_1=\gamma_2=1/2$, $\gamma=p_1\gamma_1+p_2\gamma_2=2$. Let $\bold{F}_{1,1}=\bold{F}_{2,2}=\bold{G}$ of \Cref{exam: CodeDL}. Then, $n=6$, $r=3$, $\delta'=1$, $\delta=2$ as in \Cref{exam: CodeDL}. Therefore, $d_1=r-\delta'-2\gamma'+1=3-1-2\cdot(1/2)+1=2$, $d_2=r-\delta'+\delta+1=5$, $d_3=r-\delta'+\delta-2\gamma'+\gamma+1=6$. We assume $\bold{T}_{x,i}$, $x,i\in\MB{2}$, are all identical, then so are $\bold{V}_{x,i}$ and $\bold{E}_{x,y;i;1}$, $x\neq y$, $i\in\MB{2}$. Let these matrices be defined as follows:
\begin{equation*}
\bold{V}_{x,i}=\left[\begin{array}{ccc}\frac{1}{\beta^{6}-\beta^{8}}& \frac{1}{\beta^{6}-\beta^{9}} & \frac{1}{\beta^{6}-\beta^{10}} \end{array}\right]=\left[\begin{array}{ccc}\beta^{}&\beta^{10}&\beta^{8}\end{array}\right]
\end{equation*}
\begin{equation*}\textit{and } \text{ } \bold{E}_{x,y;i;1}=
\left[\begin{array}{ccc}
\frac{1}{\beta-\beta^{12}} \\
\frac{1}{\beta^{2}-\beta^{12}} \\
\frac{1}{\beta^{3}-\beta^{12}} \\
\end{array}\right]=\left[\begin{array}{ccc}
\beta^{2}\\
\beta^{8}\\
\beta^{5}\\
\end{array}\right].
\end{equation*}
For simplicity, we abbreviate $\bold{E}_{x,y;i;1}$ as $\bold{E}$.
Note that here $p_1$, $p_2$ are even; thus, the construction follows the modification described in \Cref{rem:1}. The components $\bold{A}_{x,y;i,j}$ are therefore all identical for $x,y,i,j\in\MB{2}$, $x\neq y$, and are described as follows:
\begin{equation*}\bold{A}_{x,y;i,j}=\bold{E}\bold{V}_{y,j}=\left[\begin{array}{ccc}
\beta^{3} &\beta^{12} & \beta^{10}\\
\beta^9 &\beta^{3} & \beta^{}\\
\beta^{6} &1& \beta^{13}
\end{array}
\right].
\end{equation*}
Then, the generator matrix is given in (\ref{eqn: exam2}).

\begin{figure*}[!t]
\normalsize
\setcounter{equation}{11}
\begin{equation}\label{eqn: exam2}\small
\left[
\begin{array}{ccc|ccc|ccc|ccc|ccc|ccc|ccc|ccc}
1 & 0 & 0 & \beta^{5} &\beta^{12} & \beta^{7} & 0 & 0 & 0 & \beta^{13} &\beta^{9} & \beta^{3} & 0 & 0 & 0 & \beta^{3} &\beta^{12} & \beta^{10} & 0 & 0 & 0 & \beta^{3} &\beta^{12} & \beta^{10}\\
0 & 1 & 0 & 1 &\beta^{4} & \beta^{11} & 0 & 0 & 0 & \beta^{10} &\beta^{6} & 1  & 0 & 0 & 0 & \beta^9 &\beta^{3} & \beta^{} & 0 & 0 & 0 & \beta^9 &\beta^{3} & \beta^{}\\
0 & 0 & 1 & \beta^{2} &\beta^{14} & \beta^{3} & 0 & 0 & 0 & \beta^{14} &\beta^{10} & \beta^{4}  & 0 & 0 & 0 & \beta^{6} &1& \beta^{13} & 0 & 0 & 0 & \beta^{6} &1& \beta^{13}\\
\hline
0 & 0 & 0 & \beta^{13} &\beta^{9} & \beta^{3} & 1 & 0 & 0 & \beta^{5} &\beta^{12} & \beta^{7}  & 0 & 0 & 0 & \beta^{3} &\beta^{12} & \beta^{10} & 0 & 0 & 0 & \beta^{3} &\beta^{12} & \beta^{10}\\
0 & 0 & 0 & \beta^{10} &\beta^{6} & 1 & 0 & 1 & 0 & 1 &\beta^{4} & \beta^{11}  & 0 & 0 & 0 & \beta^9 &\beta^{3} & \beta^{} & 0 & 0 & 0 & \beta^9 &\beta^{3} & \beta^{}\\
0 & 0 & 0 & \beta^{14} &\beta^{10} & \beta^{4} & 0 & 0 & 1 & \beta^{2} &\beta^{14} & \beta^{3}  & 0 & 0 & 0 & \beta^{6} &1& \beta^{13} & 0 & 0 & 0 & \beta^{6} &1& \beta^{13}\\
\hline
 0 & 0 & 0 & \beta^{3} &\beta^{12} & \beta^{10} & 0 & 0 & 0 & \beta^{3} &\beta^{12} & \beta^{10} & 1 & 0 & 0 & \beta^{5} &\beta^{12} & \beta^{7} & 0 & 0 & 0 & \beta^{13} &\beta^{9} & \beta^{3} \\
0 & 0 & 0 & \beta^9 &\beta^{3} & \beta^{} & 0 & 0 & 0 & \beta^9 &\beta^{3} & \beta^{} & 0 & 1 & 0 & 1 &\beta^{4} & \beta^{11} & 0 & 0 & 0 & \beta^{10} &\beta^{6} & 1 \\
0 & 0 & 0 & \beta^{6} &1& \beta^{13} & 0 & 0 & 0 & \beta^{6} &1& \beta^{13} & 0 & 0 & 1 & \beta^{2} &\beta^{14} & \beta^{3} & 0 & 0 & 0 & \beta^{14} &\beta^{10} & \beta^{4} \\
\hline
0 & 0 & 0 & \beta^{3} &\beta^{12} & \beta^{10} & 0 & 0 & 0 & \beta^{3} &\beta^{12} & \beta^{10} & 0 & 0 & 0 & \beta^{13} &\beta^{9} & \beta^{3} & 1 & 0 & 0 & \beta^{5} &\beta^{12} & \beta^{7} \\
0 & 0 & 0 & \beta^9 &\beta^{3} & \beta^{} & 0 & 0 & 0 & \beta^9 &\beta^{3} & \beta^{} & 0 & 0 & 0 & \beta^{10} &\beta^{6} & 1 & 0 & 1 & 0 & 1 &\beta^{4} & \beta^{11} \\
0 & 0 & 0 & \beta^{6} &1& \beta^{13} & 0 & 0 & 0 & \beta^{6} &1& \beta^{13} & 0 & 0 & 0 & \beta^{14} &\beta^{10} & \beta^{4} & 0 & 0 & 1 & \beta^{2} &\beta^{14} & \beta^{3} \\
\end{array}\right].
\end{equation}
\setcounter{equation}{12}
\hrulefill
\end{figure*}

Note that the decoding process based on local access and global access have already been introduced in \Cref{exam: CodeDL}. Thus, we only focus on decoding based on the middle-level access in this example. Suppose $\bold{m}_{1,1}=(1,\beta^{},\beta^{2})$, $\bold{m}_{1,2}=(\beta,1,0)$, $\bold{m}_{2,1}=(\beta^{2},0,\beta^{})$, $\bold{m}_{2,2}=(0,\beta^{},1)$. Then, $\bold{c}_{1,1}=(1,\beta^{},\beta^{2},\beta^{12},\beta^{14},\beta^{12})$, $\bold{c}_{1,2}=(\beta^{},1,0,\beta^{9},\beta^{14},\beta)$. 

Suppose there are $3$ erasures in $\bold{c}_{1,1}$ so that $\bold{c}'_{1,1}=(e_1,\beta,\beta^2,e_2,e_3,\beta^{12})$, where $e_1,e_2,e_3$ represent the three erased symbols. Suppose $\bold{c}_{1,2}$ is successfully corrected by local access. Then, codeword $\bold{c}_{1,1}$ is correctable through middle-level access, i.e., by operating on $\bold{c}'_{1,1}$ and $\bold{c}_{1,2}$. 

First, from $\bold{c}_{1,2}=(\beta^{},1,0,\beta^{9},\beta^{14},\beta)$, we know that $\bold{m}_{1,2}=(\beta^{},1,0)$. Following the proof of \Cref{theo: HLdis}, we know that $(\beta^{9},\beta^{14},\beta)=\bold{m}_{1,2}\bold{A}_{1,1;1,2}+\bold{y}_{1,2}\bold{U}_{1,2}+\bold{z}_{1,2}\bold{V}_{1,2}$. Here, $\bold{y}_{1,1}=\bold{m}_{1,1}\bold{B}_{1,1;1,2}$, $\bold{z}_{1,2}=(\bold{m}_{2,1}+\bold{m}_{2,2})\bold{E}=\bold{z}_{1,1}$. Then, $\bold{y}_{1,2}$ and $\bold{z}_{1,2}$ can be computed as $\bold{y}_{1,2}=(\beta^{11})$, $\bold{z}_{1,2}=(\beta^4)$. Therefore, $\bold{z}_{1,1}\bold{V}_{1,1}+\bold{m}_{1,2}\bold{A}_{1,1;2,1}=\bold{z}_{1,2}\bold{V}_{1,1}+\bold{m}_{1,2}\bold{A}_{1,1;2,1}=(\beta^5,\beta^{14},\beta^{12})+(\beta^{11},\beta^7,\beta)=(\beta^3,\beta,\beta^{13})$. 

Let $\tilde{\bold{c}}_{1,1}=\bold{c}'_{1,1}-(0,0,0,\beta^3,\beta,\beta^{13})=(e'_1,\beta,\beta^2,e'_2,e'_3,\beta)$. We obtain $(e'_1,e'_2,e'_3)=(1,\beta^{10},\beta^{7})$ by solving $\bold{H}^{\Tx{G}}_1\tilde{\bold{c}}_{1,1}^{\Tx{T}}=(0,0,0,e^{11})^{\Tx{T}}$, where $\bold{H}^{\Tx{G}}_1$ is specified in \Cref{exam: CodeDL}. Therefore, $e_1=e'_1=1$, $e_2=e'_2+\beta^3=\beta^{12}$, $e_3=e'_3+\beta=\beta^{14}$. We have successfully decoded $\bold{c}_{1,1}$.
\end{exam}

\section{Scalability, Heterogeneity, and Flexibility}
\label{section: scalability and flexibility in cloud storage}

In \Cref{section: codes for multi-level access}, we have presented a construction of codes with hierarchical locality for cloud storage, which enables the system to offer multi-level access. However, multi-level accessibility is not the only property that is considered in practical cloud storage applications. In this section, we therefore discuss scalability, heterogeneity, and flexibility of our construction, which are pivotal particularly in dynamic cloud storage. Although our discussion is restricted to cloud storage, the properties of heterogeneity and flexibility are also of practical importance in non-volatile memories.

\subsection{Scalability}
\label{subsection: scalability}
As discussed in \Cref{section: introduction}, scalability refers to the capability of expanding the backbone network to accommodate additional workload without rebuilding the entire infrastructure. More specifically, when a new local cloud is added to the existing configuration, computing a completely different generator matrix resulting in changing all the encoding-decoding components in the system is very costly. The ideal scenario is that adding a new local cloud does not change the encoding-decoding components of the already-existing, local clouds. 

We show that our construction naturally achieves this goal. Observe that in \Cref{cons: CRScons}, the components $\bold{A}_{x,x}$, $\bold{U}_x$, $\bold{B}_{x,i}$, $i\in\MB{p}\setminus\LB{x}$ are built locally. Suppose cloud $p+1$ is added into a double-level configuration adopting \Cref{cons: CRScons}. The following steps will only result in adding some columns and rows to the original $\bold{G}$ without changing the existing ones:

\begin{enumerate}
\item \emph{Parameter Selection:} Local cloud $p+1$ chooses its local parameters $\bold{A}_{p+1,p+1}$, $\bold{U}_{p+1}$, $\bold{B}_{p+1,i}$, $i\in\MB{p}$, and local cloud $i$ chooses the additional local parameters $\bold{B}_{i,p+1}$;
\item \emph{Information Exchange:} Local cloud $p+1$ sends $\bold{m}_{p+1}\bold{B}_{p+1,i}$ to the central cloud, and local cloud $i$ sends $\bold{m}_i\bold{B}_{i,p+1}$ to the central cloud;
\item \emph{Information Exchange:} The central cloud forwards $\bold{m}_{p+1}\bold{B}_{p+1,i}$ to local cloud $i$, and sends $\bold{y}_{p+1}=\sum\nolimits_{i\in\MB{p}} \bold{m}_i\bold{B}_{i,p+1}$ to local cloud $p+1$; 
\item \emph{Update:} Local cloud $p+1$ computes its finalized parity-check symbols $\bold{m}_{p+1}\bold{A}_{p+1,p+1}+\bold{y}_{p+1}\bold{U}_{p+1}$, and local cloud $i$ adds $\bold{m}_{p+1}\bold{B}_{p+1,i}$ to its current parity symbols.
\end{enumerate}

Note that although the local erasure-correction capability of a local cloud does not change, the global erasure-correction capability of each local cloud increases by $\delta_{p+1}$ after adding the new local cloud $p+1$ into the system.

\subsection{Heterogeneity}
\label{subsection: heterogeneity}
While codes with identical data length and locality have been intensively studied, heterogeneity has become increasingly important in real world applications, especially in cloud storage. There are typically two forms of heterogeneity: the heterogeneity of the network structure, and unequal usage rates (according to how hot the data stored are) of local components. It is reasonable to assume a heterogeneous structure since components connected to a larger network are typically geographically separated and they often store data from unrelated sources. Heterogeneous networks naturally require codes with different local code lengths and nonidentical data lengths, corresponding to flexible $n_x$ and $k_x$ in our construction, respectively. Unequal protection of data, corresponding to flexible $r_x$ and $\delta_x$, also has received increasing attention in recent years. This observation is reasonable since the usage rate of the data is not necessarily identical. Clouds storing hot data (data with higher usage rate and more time urgency) should receive more local protection than those store cold data.

Although the examples we presented in \Cref{section: codes for multi-level access} have identical local parameters among all the clouds for simplicity, \Cref{cons: CRScons} and \Cref{cons: ConsHL} do not impose such restrictions, and they are actually suitable for heterogeneous configuration. 

\begin{exam}\label{exam: hetero} Here, we build on \Cref{exam: defi2} and we use the same parameters. In this example, $n_x$, $k_x$, $\delta_x$ are not identical for all $x$.
 Let $(\delta_{1,1},\delta_{1,2})=(1,1)$; thus, $\delta_1=1+1=2$. Let $(\delta_{2,1},\delta_{2,2})=(1,1)$; thus, $\delta_2=1+1=2$. Let $(\delta_{3,1},\delta_{3,2},\delta_{3,3},\delta_{3,4})=(1,2,1,1)$; thus, $\delta_3=1+2+1+1=5$.

Let $\gamma_1=1$ and $\gamma_2=\gamma_3=1/2$; thus, $\gamma=2\cdot (1)+2\cdot(1/2)+4\cdot (1/2)=5$.

Then, $d_{1,1,1}=r_{1,1}-\delta_{1,1}-2\gamma_1+1=4-1-2\cdot 1+1=2$; $d_{2,1,1}=r_{1,1}-\delta_{1,1}+\delta_1+1=4-1+2+1=6$; $d_{3,1,1}=r_{1,1}-\delta_{1,1}+\delta_1-p_1\gamma_1+\gamma+1=4-1+2-2\cdot 1+5+1=9$. The rest of the parameters can be obtained in a similar fashion, and we then specify $\bold{D}$ as follows:
\begin{equation}
\bold{D}=\left[\begin{array}{cc|cc|cccc}
2 &3 & 2 & 2 & 2 & 2 & 2 & 2\\ 
\hline
6 &7 & 5 & 5 & 8 & 8 & 8 & 8\\ 
\hline
9 &10 & 9 & 9 & 11 & 11 & 11 & 11 
\end{array}\right].
\end{equation}

According to \Cref{cons: ConsHL}, one can construct an $(\bold{n},\bold{k},\bold{D},p_0,\bold{p})_q$-code with the parameters specified previously.
\end{exam}

\subsection{Flexibility}
\label{subsection: flexibility}
The concept of flexibility has been originally proposed and investigated for dynamic cloud storage in \cite{martnez2018universal}. In a dynamic cloud storage system, the usage rate of a piece of data is not likely to remain unchanged. When the data stored in a local cloud become hot, splitting the local cloud into two smaller clouds effectively reduces the latency. However, this action should be done without reducing the erasure-correction capability of the rest of the system or changing the remaining components.

Take \Cref{cons: CRScons} as an example, if the data stored in local cloud $1$ becomes unexpectedly hot, then the following procedure splits it into two separate clouds $1^{\Tx{a}}$ and $1^{\Tx{b}}$:

\begin{enumerate}
\item Select the desired local parameters $(k_1^{\Tx{a}},r_1^{\Tx{a}},\delta_1^{\Tx{a}})$ and $(k_1^{\Tx{b}},r_1^{\Tx{b}},\delta_1^{\Tx{b}})$ for clouds $1^{\Tx{a}}$ and $1^{\Tx{b}}$, respectively, such that $k_1^{\Tx{a}}+k_1^{\Tx{b}}=k_1$, $r_1^{\Tx{a}}+r_1^{\Tx{b}}=r_1$, $\delta_1^{\Tx{a}}+\delta_1^{\Tx{b}}=\delta_1$, and 
\begin{align*}
\bold{A}_{1^{\Tx{a}},1^{\Tx{a}}}&=\bold{A}_{1,1}\MB{1:k_1^{\Tx{a}}, 1: r_1^{\Tx{a}}},\\
\bold{B}_{1^{\Tx{b}},1^{\Tx{a}}}&=\bold{A}_{1,1}\MB{k_1^{\Tx{a}}+1:k_1, 1: \delta_1^{\Tx{a}}},\\
\bold{A}_{1^{\Tx{b}},1^{\Tx{b}}}&=\bold{A}_{1,1}\MB{k_1^{\Tx{a}}+1:k_1, ,r_1^{\Tx{a}}+1:r_1},\\
\bold{B}_{1^{\Tx{a}},1^{\Tx{b}}}&=\bold{A}_{1,1}\MB{1: k_1^{\Tx{a}},r_1^{\Tx{a}}+1:r_1^{\Tx{a}}+\delta_1^{\Tx{b}}},\\
\bold{B}_{1^{\Tx{a}},i}&=\bold{B}_{1,i}\MB{1: k_1^{\Tx{a}},1:r_i},\forall 2\leq i\leq p,\\
\bold{B}_{1^{\Tx{b}},i}&=\bold{B}_{1,i}\MB{k_1^{\Tx{a}}+1:k_1,1:r_i},\forall 2\leq i\leq p,\\
\bold{B}_{i,1^{\Tx{a}}}&=\bold{B}_{1,i}\MB{1:k_i,1: \delta_1^{\Tx{a}}},\forall 2\leq i\leq p,\\
\bold{B}_{i,1^{\Tx{b}}}&=\bold{B}_{1,i}\MB{1:k_i,\delta_1^{\Tx{a}}+1:\delta_1},\forall 2\leq i\leq p,\\
\bold{U}_1^{\Tx{a}}&=\bold{U}_1\MB{1:\delta_1^{\Tx{a}},1:r_1^{\Tx{a}}},\\
\bold{U}_1^{\Tx{b}}&=\bold{U}_1\MB{\delta_1^{\Tx{a}}+1:\delta_1,r_1^{\Tx{a}}+1:r_1^{\Tx{b}}};
\end{align*}
\item Compute $\bold{y}_1$ by solving the equation $\bold{y}_1\bold{U}_1=\bold{c}_1-\bold{m}_1\bold{A}_{1,1}$, where $\bold{y}_i$, $i\in\MB{p}$, are described in the proof of \Cref{lemma: DLcodedis}. Find $\bold{y}_1^{\Tx{a}}\in \textup{GF}(q)^{\delta_1^{\Tx{a}}}$, $\bold{y}_1^{\Tx{b}}\in \textup{GF}(q)^{\delta_1^{\Tx{b}}}$ such that $\bold{y}_1=\MB{\bold{y}_1^{\Tx{a}},\bold{y}_1^{\Tx{b}}}$;
\item Compute $\bold{c}_1^{\Tx{a}}=\MB{\bold{m}_1^{\Tx{a}},\bold{m}_1^{\Tx{a}}\bold{A}_{1^{\Tx{a}},1^{\Tx{a}}}+\SB{\bold{m}_1^{\Tx{b}}\bold{B}_{1^{\Tx{b}},1^{\Tx{a}}}+\bold{y}_1^{\Tx{a}}}\bold{U}_1^{\Tx{a}}}$, and $\bold{c}_1^{\Tx{b}}=\MB{\bold{m}_1^{\Tx{b}},\bold{m}_1^{\Tx{b}}\bold{A}_{1^{\Tx{b}},1^{\Tx{b}}}+\SB{\bold{m}_1^{\Tx{a}}\bold{B}_{1^{\Tx{a}},1^{\Tx{b}}}+\bold{y}_1^{\Tx{b}}}\bold{U}_1^{\Tx{b}}}$.
\end{enumerate}
Note that the matrix $\bold{B}_{1,i}$ is vertically split into $\bold{B}_{1^{\textup{a}},i}$ and $\bold{B}_{1^{\textup{b}},i}$, while $\bold{B}_{i,1}$ is horizontally split into $\bold{B}_{i,1^{\textup{a}}}$ and $\bold{B}_{i,1^{\textup{b}}}$, for all $2\leq i\leq p$. Therefore, it is obvious that $\bold{m}_1\bold{B}_{1,i}=\bold{m}_{1^{\textup{a}}}\bold{B}_{1^{\textup{a}},i}+\bold{m}_{1^{\textup{b}}}\bold{B}_{1^{\textup{b}},i}$ and one can prove that the local codeword $\bold{c}_i$ doesn't change for $2\leq i\leq p$. Moreover, since both the local and the global parity check matrices for each non-split local cloud remain unchanged, the local and global erasure capability are not affected according to \Cref{lemma: DLcodedis}. Furthermore, one can prove that the local codewords stored in the new clouds $1^{\Tx{a}}$ and $1^{\Tx{b}}$ such that they are capable of correcting $(r_1^{\Tx{a}}-\delta_1^{\Tx{a}})$ and $(r_1^{\Tx{b}}-\delta_1^{\Tx{b}})$ local erasures, respectively.

\section{Conclusion}
\label{section: conclusion}
Multi-level accessible codes have been shown to be beneficial for cloud storage. While the previous literature works was typically focused on double-level accessible codes and their erasure-correction capabilities, in this paper, we focus on codes with hierarchical locality and additional properties motivated by their practical importance. We proposed a CRS-based code on a finite field with size that grows linearly with the maximum local codelength. We showed that our construction achieves scalability, heterogeneity and flexibility, which are important in dynamic cloud storage. 


\section*{Acknowledgment}
This work has received funding from NSF under the grants CCF-BSF 1718389 and CCF 1717602.

\bibliography{ref}
\bibliographystyle{IEEEtran}

\end{document}